\documentclass[11pt,letterpaper]{amsart}

\usepackage{amsmath,amsfonts,amsthm,amssymb,stmaryrd,bm,cite,enumerate}
\theoremstyle{plain}

\numberwithin{equation}{section}

\newtheorem{thm}{Theorem}[section]
\newtheorem{lem}[thm]{Lemma}

\theoremstyle{definition}
\newtheorem{example}{Example}

\allowdisplaybreaks  

\newcommand{\trace}{\mathrm{tr\,}}
\newcommand{\rmin}{\mathrm{In\,}}
\newcommand{\rmob}{\mathrm{Ob\,}}
\newcommand{\ityes}{\textit{yes}}      
\newcommand{\itno}{\textit{no}}

\newcommand{\escript}{\mathcal{E}}
\newcommand{\hscript}{\mathcal{H}}
\newcommand{\iscript}{\mathcal{I}}
\newcommand{\jscript}{\mathcal{J}}
\newcommand{\lscript}{\mathcal{L}}
\newcommand{\oscript}{\mathcal{O}}
\newcommand{\sscript}{\mathcal{S}}
\newcommand{\hscriptbar}{\overline{\hscript}}
\newcommand{\iscriptbar}{\overline{\iscript}}
\newcommand{\jscriptbar}{\overline{\jscript}}
\newcommand{\hscripthat}{\widehat{\hscript}}
\newcommand{\iscripthat}{\widehat{\iscript}}
\newcommand{\jscripthat}{\widehat{\jscript}}

\newcommand{\brac}[1]{\left\{#1\right\}}
\newcommand{\paren}[1]{\left(#1\right)}
\newcommand{\sqbrac}[1]{\left[#1\right]}

\newcommand{\doubleab}[1]{\left|\left|#1\right|\right|}
\newcommand{\elbows}[1]{{\left\langle#1\right\rangle}}
\newcommand{\ket}[1]{{\left|#1\right>}}
\newcommand{\bra}[1]{{\left<#1\right|}}

\begin{document}

\title{QUANTUM CHANNEL CONDITIONING AND MEASUREMENT MODELS}

\author{Stan Gudder}
\address{Department of Mathematics, 
University of Denver, Denver, Colorado 80208}
\email{sgudder@du.edu}
\date{}
\maketitle

\begin{abstract}
If $H_1$ and $H_2$ are finite-dimensional Hilbert spaces, a channel from $H_1$ to $H_2$ is a completely positive, linear map $\iscript$ that takes the set of states $\sscript (H_1)$ for $H_1$ to the set of states $\sscript (H_2)$ for $H_2$. Corresponding to $\iscript$ there is a unique dual map $\iscript ^*$ from the set of effects $\escript (H_2)$ for $H_2$ to the set of effects $\escript (H_1)$ for $H_1$. We call $\iscript ^*(b)$ the effect $b$ conditioned by $\iscript$ and the set $\iscript ^c=\iscript ^*\paren{\escript (H_2)}$ the conditioned set of $\iscript$. We point out that $\iscript ^c$ is a convex subeffect algebra of the effect algebra $\escript (H_1)$. We extend this definition to the conditioning
$\iscript ^*(B)$ for an observable $B$ on $H_2$ and say that an observable $A$ is in $\iscript ^c$ if $A=\iscript ^*(B)$ for some observable $B$. We show that $\iscript ^c$ is closed under post-processing and taking parts. We also define the conditioning of instruments by channels. These concepts are illustrated using examples of Holevo instruments and channels. We next discuss measurement models and their corresponding observables and instruments. We show that calculations can be simplified by employing Kraus and Holevo separable channels. Such channels allow one to separate the components of a tensor product.
\end{abstract}

\section{Basic Concepts and Definitions}  
All the Hilbert spaces in this article are assumed to be finite-dimensional. For a Hilbert space $H$, we denote the set of linear operators on $H$ by $\lscript (H)$ and the set of self-adjoint operators on $H$ by $\lscript _S(H)$. The zero and identity operators are designated by $0$, $I$, respectively. When the Hilbert space needs to be specified, we write $I_H$ instead of $I$. An operator $a\in\lscript _S(H)$ that satisfies
$0\le a\le I$ is called an \textit{effect}. The set of effects on $H$ is denoted by $\escript (H)$ and is called an \textit{effect algebra}
\cite{bgl95,blm96,blpy16,gg01,gn01}. We consider an effect as a two-valued \ityes --\itno\  experiment and when the value \ityes\  is obtained, then $a$ \textit{occurs} \cite{hz12,kra83,nc00}. A function $f\colon\escript (H_1)\to\escript (H_2)$ between effect algebras is
\textit{additive} if $a,b\in\escript (H_1)$ satisfy $a+b\le I_{H_1}$ then $f(a)+f(b)\le I_{H_2}$ and $f(a+b)=f(a)+f(b)$. We say that $f$ is a
\textit{morphism} if $f$ is additive and $f(I_{H_1})=I_{H_2}$ \cite{hz12,nc00}.

If $\Omega _A$ is a finite set, then a set of effects $A=\brac{A_x\colon x\in\Omega _A}\subseteq\escript (H)$ that satisfies
$\sum\limits _{x\in\Omega _A}A_x=I$ is called an \textit{observable}. The set of observables on $H$ is denoted by $\rmob (H)$. We call
$\Omega _A$ the \textit{outcome space} for $A$ and when $\Delta\subseteq\Omega _A$, the map $A(\Delta )=\sum\limits _{x\in\Delta}A_x$ is an
\textit{effect-valued measure} (or \textit{positive operator-valued measure} (POVM)) \cite{bgl95,blm96,blpy16,dl70,hz12,nc00}. We think of $A_x$ as the effect that occurs when a measurement of $A$ results in the outcome $x$. A \textit{state} on $H$ is a positive operator $\rho\in\lscript _S(H)$ with trace $\trace (\rho )=1$ \cite{blpy16,hz12,nc00} and the set of states on $H$ is denoted by $\sscript (H)$. A state specifies the initial condition of a quantum system and determines the probabilities of measurement results in the following sense. If $\rho\in\sscript (H)$, $a\in\escript (H)$, then $0\le\trace (\rho a)\le 1$ and we call $P_\rho (a)=\trace (\rho a)$ the
$\rho$-\textit{probability that} $a$ \textit{occurs}. This is called the \textit{Born rule} \cite{hz12,nc00}. Of course, $P_\rho (0)=0$ and $P_\rho (I)=1$ for all
$\rho\in\sscript (H)$ so $0$ never occurs and $I$ always occurs. For $A\in\rmob (H)$, $\rho\in\sscript (H)$, the probability distribution on
$\Omega _A$ given by
\begin{equation*}
\Phi _\rho ^A(\Delta )=\trace\sqbrac{\rho A(\Delta )}
\end{equation*}
is called the $\rho$-\textit{distribution} of $A$ in the state $\rho$ \cite{bgl95,blm96,blpy16,dl70,hz12,nc00} A set of effects
$\brac{a_i\colon i=1,2,\ldots ,n}\subset\escript (H)$ \textit{coexist} (are \textit{compatible}) if there exits an observable $A\in\rmob (H)$ such that
$a_i=A(\Delta _i)$, $\Delta _i\subseteq\Omega _A$, $i=1,2,\ldots ,n$. It can be shown that if a set of effects mutually commute, then they coexist but the converse does not hold \cite{bkmpt22,hz12,mf23,nc00}. We call $A$ a \textit{joint observable} for
$\brac{a_i\colon i=1,2,\ldots ,n}$. Clearly, a finite set of effects coexist if and only if they are measured by a single observable.

An observable $B$ is \textit{part} of an obsservable $H$ if there exists a surjection $f\colon\Omega _A\to\Omega _B$ such that
\begin{equation*}
B_y=A\sqbrac{f^{-1}(y)}=\sum\brac{A_x\colon f(x)=y}
\end{equation*}
for all $y\in\Omega _B$ and we write $B=f(A)$ \cite{gud220,gud21}. We see that $B$ is obtained by piecing together effects from $A$. If $A\in\rmob (H)$, $\Omega _B$ is a finite set and $\lambda _{xy}\in\sqbrac{0,1}$, $x\in\Omega _A$, $y\in\Omega _B$ satisfy
$\sum\limits _{y\in\Omega _B}\lambda _{xy}=1$ for all $x\in\Omega _A$, then the observable
$B_y=\sum\limits _{x\in\Omega _A}\lambda _{xy}A_x$ is called a \textit{post-processing} of $A$ \cite{gud21,hz12}. We think of $B$ as the observable obtained by first measuring $A$ and then processing $A$ with a classical probability distribution $\lambda _{xy}$. An observable $A$ is a
\textit{bi-observable} if $\Omega _A=\Omega _1\times\Omega _2$ in which case we write 
\begin{equation*}
A=\brac{A_{xy}\colon (x,y)\in\Omega _1\times\Omega _2}
\end{equation*}
If $A\in\rmob (H)$ is a bi-observable, we define the \textit{marginals of} $A$ to be the observables $A^1,A^2\in\rmob (H)$ given by $A_x^1=\sum\limits _{y\in\Omega _2}A_{xy}$, $A_y^2=\sum\limits _{x\in\Omega _1}A_{xy}$ for all $x\in\Omega _1$, $y\in\Omega _2$. Notice that $A^1,A^2$ are parts of $A$ with $A^1=f_1(A)$, $A^2=f_2(A)$ where $f_1(x,y)=x$, $f_2(x,y)=y$ for all $x\in\Omega _1$, $y\in\Omega _2$. Two observables $A,B\in\rmob (H)$ \textit{coexist} (are \textit{compatible}) if there exists a
\textit{joint bi-observable} $C\in\rmob(H)$ such that $A=C^1$ and $B=C^2$. We see that $A,B$ coexist if there is a joint bi-observable that simultaneously measures them both. In particular, the marginals of a bi-observable coexist. It is not hard to show that the effects for two coexisting observables also coexist \cite{bkmpt22,hz12,mf23}

\begin{lem}    
\label{lem11}
{\rm{(i)}}\enspace A post-processing of a post-processing of $A\in\rmob (H)$ is a post-processing of $A$.
{\rm{(ii)}}\enspace A part of a part of $A\in\rmob (H)$ is a part of $A$.
\end{lem}
\begin{proof}
(i)\enspace Suppose $B_y=\sum\limits _{x\in\Omega _A}\lambda _{xy}A_x$ is a post-processing of $A$ and $\mu _{yz}\in\sqbrac{0,1}$ with
$\sum\limits _z\mu _{yz}=1$ for all $y$. Then for $\delta _{xz}=\sum\limits _{y\in\Omega _B}\mu _{yz}\lambda _{xy}$ we have the post-processing of $B$ given by
\begin{align*}
C_z&=\sum _{y\in\Omega _B}\mu _{yz}B_y=\sum _{y\in\Omega _B}\mu _{yz}\sum _{x\in\Omega _A}\lambda _{xy}A_x\\
&=\sum _{x\in\Omega _A}\paren{\sum _{y\in\Omega _B}\mu _{yz}\lambda _{xy}}A_x=\sum _x\delta _{xz}A_x
\end{align*}
Since
\begin{equation*}
\sum _{z\in\Omega _C}\delta _{xz}=\sum _{z\in\Omega _C}\sum _{y\in\Omega _B}\mu _{yz}\lambda _{xy}
=\sum _{y\in\Omega _B}\sum _{z\in\Omega _C}\mu _{yz}\lambda _{xy}=\sum _{y\in\Omega _B}\lambda _{xy}=1
\end{equation*}
We conclude that $C$ is a post-processing of $A$.
(ii)\enspace Let $B=f(A)$ be a part of $A$ and $C=g(B)$ be a part of $B$. Letting $h\colon\Omega _A\to\Omega _C$ be given by $h=g\circ f$ we conclude that $h$ is a surjection so $C=h(A)$ is a part of $A$.
\end{proof}

An \textit{operation} $\iscript$ from $H$ to $H_1$ is a trace non-increasing, completely positive, linear map
$\iscript\colon\lscript (H)\to\lscript (H_1)$ \cite{bgl95,blm96,blpy16,dl70}. We denote the set of operations from $H$ to $H_1$ by
$\oscript (H,H_1)$. If $\iscript\in\oscript (H,H_1)$ preserves the trace, we call $\iscript$ a \textit{channel} \cite{blpy16,gud21,hz12}. We think of a channel
$\iscript\in\oscript (H,H_1)$ as an interaction between a quantum system described by $H$ and a quantum system described by $H_1$. It can be shown that every $\iscript\in\oscript (H,H_1)$ has the form $\iscript (B)=\sum\limits _{i=1}^nK_iBK_i^*$ where $K_i\colon H\to H_1$ are linear operators satisfying $\sum\limits _{i=1}^nK_i^*K_i\le I_H$ \cite{gud21,kra83,nc00}. The operators $K_i$ are called \textit{Kraus operators} for $\iscript$. It is easy to show that $\iscript$ is a channel if and only if $\sum\limits _{i=1}^mK_i^*K_i=I_H$. We say that $\iscript\in\oscript (H,H_1)$ \textit{measures} an effect $a\in\escript (H)$ if $\trace\sqbrac{\iscript (\rho )}=\trace (\rho a)$ for all $\rho\in\sscript (H)$. We shall see that $\iscript$ measures a unique effect.

An \textit{instrument} $\iscript$ from $H$ to $H_1$ is a finite set of operations $\iscript =\brac{\iscript _x\colon x\in\Omega _\iscript}\subseteq\oscript (H,H_1)$ such that
$\iscriptbar=\sum\limits _{x\in\Omega _\iscript}\iscript _x$ is a channel \cite{gud220,gud22,gud123}. We call $\Omega _\iscript$ the \textit{outcome space} for $\iscript$ and denote the set of instruments from $H$ to $H_1$ by $\rmin (H,H_1)$. If $H=H_1$, we write $\iscript (H)$ for $\iscript (H,H)$. For $\Delta\subseteq\Omega _\iscript$ we write $\iscript (\Delta )=\sum\brac{\iscript _x\colon x\in\Delta}$ and call $\iscript$ an \textit{operation-valued measure} \cite{bgl95,blm96,blpy16,dl70,gud21}. If
$\iscript\in\rmin (H,H_1)$, $\rho\in\sscript (H)$, the $\rho$-\textit{distribution of} $\iscript$ is the probability measure on $\Omega _\iscript$ given by 
\begin{equation*}
\Phi _\rho ^\iscript (\Delta )=\trace\sqbrac{\iscript (\Delta )(\rho )}=\sum _{x\in\Delta}\trace\sqbrac{\iscript _x(\rho )}
\end{equation*}
An instrument $\iscript\in\rmin (H,H_1)$ \textit{measures} a unique observable $\iscripthat\in\rmob (H)$ given by
$\Omega _{\iscripthat}=\Omega _\iscript$ where $\trace (\rho\iscripthat _x)=\trace\sqbrac{\iscript _x(\rho )}$ for all $\rho\in\sscript (H)$. The instrument $\iscript$ provides more information than its measured observable $\iscripthat$. In fact, $\iscript$ provides more information than its measured observable $\iscripthat$. In fact, if $\iscript$ is measured with resulting outcome $x$ and effect $\iscript _x$ when the system is in state $\rho$, then the \textit{updated state} becomes
\begin{equation*}
\iscript _x(\rho )'=\frac{\iscript _x(\rho )}{\trace\sqbrac{\iscript _x(\rho )}}
\end{equation*}
whenever $\trace\sqbrac{\iscript _x(\rho )}\ne 0$. As with observables, we define a \textit{bi-instrument} $\iscript\in\rmin (H,H_1)$ to have
$\Omega _\iscript =\Omega _1\times\Omega _2$ and write $\iscript =\brac{\iscript _{xy}\colon (x,y)\in\Omega _1\times\Omega _2}$. If
$\iscript\in\rmin (H,H_1)$ is a \textit{bi-instrument} we define the \textit{marginals} $\iscript ^1,\iscript ^2\in\rmin (H,H_1)$ by
$\iscript _x^1=\sum\limits _{y\in\Omega _2}\iscript _{xy}$, $\iscript _y^2=\sum\limits _{x\in\Omega _1}\iscript _{xy}$. One can define coexistence of instruments \cite{bkmpt22,mf23}, but we shall not need that here.

\section{Conditioning}  
In order to discuss conditioning relative to a channel $\iscript$ we need to develop the concept of a dual map $\iscript ^*$ \cite{gud22,gud123,gud24}.

\begin{thm}    
\label{thm21}
For $\iscript\in\oscript (H_1,H_2)$ the following statements hold.
{\rm{(i)}}\enspace There exists a unique \textit{dual map} $\iscript ^*\colon\escript (H_2)\to\escript (H_1)$ satisfying
$\trace\sqbrac{\rho\iscript ^*(a)}=\trace\sqbrac{\iscript (\rho )a|}$ for all $\rho\in\sscript (H_1)$, $a\in\escript (H_2)$.
{\rm{(ii)}}\enspace $\iscript ^*$ is additive.
{\rm{(iii)}}\enspace $\iscript ^*$ is a morphism if and only is $\iscript$ is a channel.
\end{thm}
\begin{proof}
(i)\enspace $\iscript$ has a Kraus decomposition $\iscript (\rho )=\sum K_i\rho K_i^*$ for all $\rho\in\sscript (H_1)$, where 
$K_i\colon H_1\to H_2$ is a linear operator, $i=1,2,\ldots ,n$ with $\sum K_i^*K_i\le I_{H_1}$. Defining
$\iscript ^*\colon\escript (H_2)\to\lscript (H_1)$ by $\iscript ^*(a)=\sum K_i^*aK_i$ we have that
\begin{align*}
\elbows{\phi ,\iscript ^*(a)\phi}&=\elbows{\phi ,\sum K_i^*aK_i\phi}=\sum\elbows{K_i\phi ,aK_i\phi}\le\sum\elbows{K_i\phi ,K_i\phi}\\
   &=\elbows{\phi ,\sum K_i^*K_i\phi}\le\elbows{\phi ,I_{H_1}\phi}=\elbows{\phi ,\phi}
\end{align*}
for all $\phi\in H_1$. Since we also have $\elbows{\phi ,\iscript ^*(a)\phi}\ge 0$ for all $\pi\in H_1$, we conclude that $0\le\iscript ^*(a)\le I_{H_1}$. Hence, $\iscript ^*(a)\in\escript (H_1)$ so $\iscript ^*\colon\escript _2\to\escript _1$. We have that $\iscript ^*$ is unique because if
$\jscript\in\escript (H_2)\to\escript (H_1)$ satisfies $\trace\sqbrac{\rho\jscript (a)}=\trace\sqbrac{\iscript (\rho )a}$ for all $\rho\in\sscript (H_1)$,
$a\in\escript (H_2)$, then $\trace\sqbrac{\rho\jscript (a)}=\trace\sqbrac{\rho\iscript ^*(a)}$ for all $\rho\in\sscript (H_1)$ so $\jscript =\iscript ^*$.
(ii)\enspace To show that $\iscript ^*$ is additive, let $a,b\in\escript (H_2)$ with $a+b\in\escript (H_2)$. For all $\rho\in\sscript (H_1)$ we obtain
\begin{align*}
\trace\sqbrac{\rho\paren{\iscript ^*(a+b)}}&=\trace\sqbrac{\iscript (\rho )(a+b)}=\trace\sqbrac{\iscript (\rho )a}+\trace\sqbrac{\iscript (\rho )(b)}\\
   &=\trace\sqbrac{\rho\iscript ^*(a)}+\trace\sqbrac{\rho\iscript ^*(b)}=\trace\sqbrac{\rho\paren{\iscript ^*(a)+\iscript ^*(b)}}
\end{align*}
Hence, $\iscript ^*(a)+\iscript ^*(b)\in\escript (H_1)$ and $\iscript ^*(a+b)=\iscript ^*(a)+\iscript ^*(b)$. We conclude that $\iscript ^*$ is additive.
(iii)\enspace If $\iscript$ is a channel, then for all $\rho\in\sscript (H_1)$ we have
\begin{equation*}
\trace (\rho I_{H_1})=1=\trace\sqbrac{\iscript (\rho )I_{H_2}}=\trace\sqbrac{\rho\iscript ^*(I_{H_2})}
\end{equation*}
Hence, $\iscript ^*(I_{H_2})=I_1$ so $\iscript ^*$ is a morphism. Conversely, if $\iscript ^*$ is a morphism then for all $\rho\in\sscript (H_1)$ we obtain
\begin{equation*}
\trace\sqbrac{\iscript (\rho )}=\trace\sqbrac{\iscript (\rho )(I_{H_2})}=\trace\sqbrac{\rho\iscript ^*(I_{H_2})}=\trace\sqbrac{\rho I_{H_1}}
   =\trace (\rho )=1
\end{equation*}
Thus, $\iscript (\rho )\in\sscript _{H_2}$ for all $\rho\in\sscript (H_1)$ so $\iscript$ is a channel.
\end{proof}

If $\iscript\in\oscript (H_1,H_2)$ and $\jscript\in\oscript (H_2,H_3)$, the \textit{sequential product of} $\iscript$ \textit{then} $\jscript$ is the operation
$\iscript\circ\jscript\in\oscript(H_1,H_3)$ given by $\iscript\circ\jscript (\rho )=\jscript\paren{\iscript (\rho )}$. In a similar way, we define
$\jscript ^*\circ\iscript ^*\colon\escript (H_3)\to\escript (H_1)$ as $(\jscript ^*\circ\iscript ^*)(b)=\iscript ^*\paren{\jscript ^*(b)}$. Since
\begin{align*}
\trace\sqbrac{\rho (\iscript\circ\jscript )^*(b)}&=\trace\sqbrac{\jscript\paren{\iscript (\rho )}b|}=\trace\sqbrac{\iscript (\rho )\jscript ^*(b)}\\
   &=\trace\sqbrac{\rho\iscript ^*\paren{\jscript ^*(b)}}=\trace\sqbrac{\rho(\jscript ^*\circ\iscript ^*)(b)}
\end{align*}
for all $\rho\in\sscript (H_1)$ we conclude that $(\iscript\circ\jscript )^*=\jscript ^*\circ\iscript ^*$.

If $\iscript\in\oscript (H_1,H_2)$ is a channel and $b\in\escript (H_2)$, we write $(b\vert\iscript )=\iscript ^*(b)$ and call $(b\vert\iscript)$ the effect $b$
\textit{conditioned by} $\iscript$ \cite{gud120,gud24}. We also write $\iscript ^C=\iscript ^*\paren{\escript (H_2)}\subseteq\escript (H_1)$ the
\textit{conditioned set} of $\iscript$. If $\rho\in\sscript (H_1)$, the $\rho$-probability of $(b\vert\iscript )$ is
\begin{equation*}
P_\rho (b\vert\iscript )=\trace\sqbrac{\rho (b\vert\iscript )}=\trace\sqbrac{\rho\iscript ^*(b)}=\trace\sqbrac{\iscript (\rho )b}
\end{equation*}
which is the $\iscript (\rho )$-probability of $b$. Since $\iscript ^*\colon\escript (H_2)\to\escript (H_1)$ is a morphism, $\iscript ^C\subseteq\escript (H_1)$ is a subeffect algebra of $\escript (H_1)$. If $a_i\in\iscript ^C$, $\lambda _i\in\sqbrac{0,1}$, $i=1,2,\ldots ,n$, $\sum\limits _i\lambda _i=1$, then $a_i=(b_i\vert\iscript )$ for
$b_i\in\escript (H_2)$ and we have
\begin{align*}
\trace\sqbrac{\rho\iscript ^*\paren{\sum\lambda _ib_i}}&=\trace\sqbrac{\iscript (\rho )\sum\lambda _ib_i}=\sum\lambda _i\trace\sqbrac{\iscript (\rho )b_i}\\
   &=\sum\lambda _i\trace\sqbrac{\rho\iscript ^*(b_i)}=\trace\sqbrac{\rho\sum\lambda _i\iscript ^*(b_i)}
\end{align*}
for all $\rho\in\sscript (H_1)$. Hence,
\begin{equation*}
\sum\lambda _ia_i=\sum\lambda _i\iscript ^*(b_i)=\iscript ^*\paren{\sum\lambda _ib_i}
\end{equation*}
We conclude that $\iscript ^*$ is affine and $\sum\lambda _ia_i\in\iscript ^C$. Therefore, $\iscript ^C$ is a convex subeffect algebra of $\escript (H_1)$
\cite{gg01,gn01,hz12}. In general, $\iscript ^*$ is not an isomorphism. In the particular case of a unitary channel $\iscript (\rho )=U\rho U^*$, where $U\colon H_1\to H_2$ is a unitary operator, then $\iscript ^*$ is an isomorphism and $\iscript ^C=\escript (H_1)$.

If $B=\brac{B_x\colon x\in\Omega _B}\in\rmob (H_2)$ and $\iscript\in\oscript (H_1,H_2)$ is a channel, then $B$ \textit{conditioned by} $\iscript$ is the observable
$(B\vert\iscript )\in\rmob (H_1)$ given by
\begin{equation*}
(B\vert\iscript )_x=\iscript ^*(B_x)=(B_x\vert\iscript )
\end{equation*}
for all $x\in\Omega _B$ \cite{gud120,gud24}. We write $\iscript ^*(B)=(B\vert\iscript )$ and $A\in\iscript ^C$ if $A=(B\vert\iscript )$ for some
$B\in\rmob (H_2)$.

\begin{lem}    
\label{lem22}
If $A\in\rmob (H_1)$, $\iscript\in\oscript (H_1,H_2)$ a channel, then $A\in\iscript ^C$ if and only if $A_x=(b_x\vert\iscript )$ for all $x\in\Omega _A$ where
$b_x\in\escript (H_2)$ and $\sum\limits _{x\in\Omega _A}b_x\le I_{H_2}$.
\end{lem}
\begin{proof}
If $A\in\iscript ^C$, then $A_x=(B_x\vert\iscript )$ for all $x\in\Omega _A$ with $B_x\in\escript (H_2)$ and $\sum\limits _{x\in\Omega _A}B_x=I_{H_2}$. Conversely, suppose $A_x=(b_x\vert\iscript )$ with $b_x\in\escript _{H_2}$ and $\sum\limits _{x\in\Omega _A}b_x\le I_{H_2}$. Let $C=I_{H_2}-\sum\limits _{x\in\Omega _A}b_x$. Then
$C\in\escript (H_2)$ and define $B_x=b_x+\tfrac{C}{n}$ where $n$ is the cardinality of $\Omega _A$. Then $B_x\ge 0$ and
\begin{equation*}
\sum B_x=\sum b_x+C=I_{H_2}
\end{equation*}
so $B=\brac{b_x\colon x\in\Omega _A}\in\rmob\paren{\escript (H_2)}$. We have that
\begin{equation*}
\iscript ^*(C)=\iscript ^*(I_{H_2})-\iscript ^*\paren{\sum b_x}=I_{H_1}-\sum\iscript ^*(b_x)=I_{H_1}-\sum A_x=0
\end{equation*}
Hence,
\begin{equation*}
\iscript ^*(B_x)=\iscript ^*(b_x)+\tfrac{1}{n}\iscript ^*(C)=\iscript ^*(b_x)=A_x
\end{equation*}
We conclude that $A=(B\vert\iscript )$ so $A\in\iscript ^C$.
\end{proof}

For $\rho\in\sscript (H_1)$, $B\in\rmob (H_2)$, the $\rho$-\textit{distribution} of $(B\vert\iscript )$ is
\begin{align*}
\Phi _\rho ^{(B\vert\iscript )}(\Delta )&=\sum _{x\in\Delta}\trace\sqbrac{\rho (B\vert\iscript )_x}=\sum _{x\in\Delta}\trace\sqbrac{\rho\iscript ^*(B_x)}\\
   &=\sum _{x\in\Delta}\trace\sqbrac{\iscript (\rho )B_x}=\Phi _{\iscript (\rho )}^B(\Delta )
\end{align*}
Hence, $\Phi _\rho ^{(B\vert\iscript )}=\Phi _{\iscript (\rho )}^B$ so the $\rho$-distribution of $(B\vert\iscript )$ is the $\iscript (\rho )$ distribution of $B$ for all
$\rho\in\sscript (H_1)$. If $\iscript\in\rmin (H_1,H_2)$, $B\in\rmob (H_2)$, then $B$ \textit{given} $\iscript$ is the bi-observable $(B\|\iscript )_{xy}(H_1)$ that satisfies
$(B\|\iscript )_{xy}=\iscript _x^*(B_y)$ for all $(x,y)\in\Omega _{(B\|\iscript )}=\Omega _\iscript\times\Omega _B$. We have that $(B\|\iscript )$ is indeed an observable because
\begin{equation*}
\sum _{x,y}(B\|\iscript )_{xy}=\sum _{x,y}\iscript _x^*(B_y)=\sum _x\iscript _x^*(I_{H_2})=\iscriptbar\,^*(I_{H_2})=I_{H_1}
\end{equation*}

\begin{lem}    
\label{lem23}
If $B\in\rmob (H_2)$, $\iscript\in\rmin (H_1,H_2)$ and $\rho\in\sscript (H_1)$ then the following statements hold.
{\rm{(i)}}\enspace $(B\|\iscript )^1=\iscripthat$ and $(B\|\iscript )^2=(B\vert\iscriptbar\,)$.
{\rm{(ii)}}\enspace 
\begin{align*}
\Phi _\rho ^{(B\|\iscript )}(\Delta )&=\trace\sqbrac{\sum _{(x,y)\in\Delta}\iscript (\rho )B_y}
\intertext{and}
\Phi _\rho ^{(B\|\iscript )}(\Delta _1\times\Delta _2)&=\trace\sqbrac{\iscript _{\Delta _1}(\rho )}\Phi _{(\iscript _{\Delta _1}(\rho ))'}^B(\Delta _2)
\end{align*}
\end{lem}
\begin{proof}
(i)\enspace For all $x\in\Omega _\iscript$ we obtain 
\begin{equation*}
(B\|\iscript )_x^1=\sum _{y\in\Omega _B}(B\|\iscript )_{xy}=\sum _{y\in\Omega _B}\iscript _x^*(B_y)=\iscript _x^*(I_{H_1})=\iscripthat _x
\end{equation*}
Hence, $(B\|\iscript )^1=\iscripthat$. For all $y\in\Omega _B$ we obtain 
\begin{equation*}
(B\|\iscript )_y^2=\sum _{x\in\Omega _\iscript}(B\|\iscript )_{xy}=\sum _{x\in\Omega _\iscript}\iscript _x^*(B_y)=\iscriptbar\,^*(B_y)=(B\vert\,\iscriptbar\,)_y
\end{equation*}
Hence, $(B\|\iscript )^2=(B\vert\,\iscriptbar\,)$.
(ii)\enspace For $\Delta\subseteq\Omega _{(B\|\iscript )}$, $\rho\in\sscript (H_1)$ we have
\begin{align*}
\Phi _C^{(B\|\iscript )}(\Delta )&=\trace\sqbrac{\rho\sum _{(xy)\in\Delta}(B\|\iscript )_{xy}}=\sum _{(x,y)\in\Delta}\trace\sqbrac{\rho (B\|\iscript )_{xy}}\\
  &=\sum _{(x,y)\in\Delta}\trace\sqbrac{\rho\iscript _x^*(B_y)}=\sum _{(x,y)\in\Delta}\trace\sqbrac{\iscript _x(\rho )B_y}\\
  &=\trace\sqbrac{\sum _{(x,y)\in\Delta}\iscript _x(\rho )B_y}
\end{align*}
If $\Delta =\Delta _1\times\Delta _2$, then the previous expression becomes
\begin{align*}
\Phi _\rho ^{(B\|\iscript )}(\Delta _1\times\Delta _2)&=\trace\sqbrac{\sum _{x\in\Delta _1}\iscript _x(\rho )\sum _{y\in\Delta _2}B_y}\\
   &=\trace\sqbrac{\iscript _{\Delta _1}(\rho )}\trace\sqbrac{\iscript _{\Delta _1}(\rho )'\sum _{y\in\Delta _2}B_y}\\
   &=\trace\sqbrac{\iscript _{\Delta _1}(\rho )}\Phi _{(\iscript _{\Delta _1}(\rho ))'}(\Delta _2)\qedhere
\end{align*}
\end{proof}

It follows from Lemma~\ref{lem23}(i) that $\iscripthat$ and $(B\vert\,\iscript\,)$ coexist with joint observable $(B\|\iscript )$ for all $B\in\rmob (H_2)$. The next shows that
$\iscript ^C$ is closed under post-processing and taking parts \cite{gud21,gud22}.

\begin{thm}    
\label{thm24}
{\rm{(i)}}\enspace If $\iscript\in\oscript (H_1,H_2)$ is a channel and $A,C\in\rmob (H_2)$ with $C$ a post-processing of $A$, then $(C\vert\iscript )$ is a post-processing of
$(A\vert\iscript )$.
{\rm{(ii)}}\enspace $\iscript ^C$ is closed under post-processing.
{\rm{(iii)}}\enspace If $\iscript\in\oscript (H_1,H_2)$ is a channel and $C\in\rmob (H_2)$ then $f\sqbrac{(C\vert\iscript )}_y=\iscript ^*\sqbrac{f(C)_y}$ for all
$y\in\Omega _{f(C)}$.
{\rm{(iv)}}\enspace $\iscript ^C$ is closed under taking parts.
\end{thm}
\begin{proof}
(i)\enspace We have that $C_y=\sum _x\lambda _{xy}A_x$ where $\lambda _{xy}\in\sqbrac{0,1}$ with $\sum _y\lambda _{xy}=1$ for all $x\in\Omega _A$. We then obtain
\begin{equation*}
(C\vert\iscript )_y=\iscript ^*(C_y)=\iscript ^*\paren{\sum _x\lambda _{xy}A_x}=\sum _x\lambda _{xy}\iscript ^*(A_x)=\sum _x\lambda _{xy}(A\vert\iscript )_x
\end{equation*}
which is a post-processing of $(A\vert\iscript )$.
(ii)\enspace If $A\in\iscript ^C$, then $A_x=\iscript ^*(B_x)$ for $B\in\rmob (H_2)$ and a post-processing of $A$  is given by $C_y=\sum _x\lambda _{xy}A_x$. But then 
\begin{equation*}
C_y=\sum _x\lambda _{xy}\iscript ^*(B_x)=\iscript ^*\paren{\sum _x\lambda _{xy}B_x}
\end{equation*}
Since $\sum _x\lambda _{xy}B_x\in\rmob (\escript _2)$, we conclude that $C\in\iscript ^C$. Hence, $\iscript ^C$ is closed under post-processing.
(iii)\enspace We have for all $y\in\Omega _{f(C)}$ that
\begin{equation*}
f\sqbrac{(C\vert\iscript )}_y=\sum _{x\in f^{-1}(y)}(C\vert\iscript )_x=\sum _{x\in f^{-1}(y)}\iscript ^*(C_x)=\iscript ^*\paren{\sum _{x\in f^{-1}(y)}C_x}=\iscript ^*\sqbrac{f(C)_y}
\end{equation*}
(iv)\enspace Suppose $A\in\iscript ^C$ so there exists a $C\in\rmob (H_2)$ such that $A=(C\vert\iscript )$. Hence, 
\begin{equation*}
A_x=(C\vert\iscript )_x=\iscript ^*(C_x)
\end{equation*}
If $B=f(A)$ is a part of $A$, then letting $D=F(C)$ and applying (iii) gives
\begin{equation*}
B_y=f(A)_y=f\sqbrac{(C\vert\iscript )}_y=\iscript ^*\sqbrac{f(C)_y}=\iscript ^*(D_y)
\end{equation*}
Hence, $B\in\iscript ^C$ so $\iscript ^C$ is closed under taking parts.
\end{proof}

If $A_i\in\rmob (H)$ with $\Omega _{A_i}=\Omega$, $i=1,2,\ldots ,n$ and $\lambda _i\in\sqbrac{0,1}$ with $\sum\limits _i\lambda _i=1$, then we call
$B=\sum\lambda _iA_i\in\rmob (H)$ given by $B_x=\sum\lambda _iA_{ix}$ an \textit{affine combination} of $\brac{A_i\colon i=1,2,\ldots ,n}$. We now show that
$\iscript ^C$ is closed under affine combinations. Suppose that $\iscript\in\rmin (H_1,H_2)$ and $A_i\in\iscript ^*(C_i)$, $C_i\in\rmob (H_2)$. Letting
$B=\sum\lambda _iA_i\in\rmob (H_2)$ be an affine combination of $\brac{A_i, i=1,2,\ldots ,n}$ we obtain
\begin{equation*}
B_x=\sum _i\lambda _i\iscript ^*(C_{ix})=\iscript ^*\paren{\sum _i\lambda _iC_{ix}}=\iscript ^*\sqbrac{\paren{\sum _i\lambda _iC_i}_x}
\end{equation*}
Since $\sum\limits _i\lambda _iC_i\in\rmob (H_2)$, we conclude that $B\in\iscript ^C$ so $\iscript ^C$ is closed under affine combinations.

If $\iscript\in\oscript (H_1,H_2)$ is a channel and $\jscript\in\rmin (H_2,H_3)$, then $\jscript$ \textit{conditioned} by $\iscript$ is the instrument
$(\jscript\vert\iscript )\in\rmin (H_1,H_3)$ given by
\begin{equation*}
(\jscript\vert\iscript )_y(\rho )=\iscript\circ\jscript _y(\rho )=\jscript _y(\iscript (\rho ))
\end{equation*}
for all $\rho\in\sscript (H_1)$ \cite{gud120,gud24}. We have for all $a\in\escript (H_3)$ that
\begin{equation*}
(\jscript\vert\iscript )_y^*=(\iscript\circ\jscript _y)^*=\jscript _y^*\circ\iscript ^*
\end{equation*}
Moreover,
\begin{equation*}
(\jscript\vert\iscript )_y^\wedge =(\jscript\vert\iscript )_y^*(I_{H_3})=\jscript _y^*\circ\iscript ^*(I_{H_3})=\iscript ^*\paren{\jscript _y^*(I_{H_3})}=\iscript ^*(\jscripthat _y)
\end{equation*}
We conclude that $(\jscript\vert\iscript )^\wedge\in\iscript ^C$. If $\iscript\in\rmin (H_1,H_2)$, $\jscript\in\rmin (H_2,H_3)$, we also call
$(\jscript\vert\iscriptbar\,)\in\rmin (H_1,H_3)$ the instrument $\jscript$ \textit{conditioned by} $\iscript$.

If $\iscript\in\rmin (H_1,H_2)$, $\jscript\in\rmin (H_2,H_3)$, then $\jscript$ \textit{given} $\iscript$ is the bi-instrument $(\jscript\|\iscript )\in\rmin (H_1,H_3)$ defined by
\begin{equation*}
(\jscript\|\iscript )_{xy}(\rho )=(\iscript _x\circ\jscript _y)(\rho )=\jscript _y\paren{\iscript _x(\rho )}
\end{equation*}
for all $\rho\in\sscript (H_1)$. The marginals of $(\jscript\|\iscript )$ become
\begin{align*}
(\jscript\|\iscript )_x^1&=\sum _y(\jscript\|\iscript )_{xy}=\sum _y\iscript _x\circ\jscript _y=\iscript _x\circ\jscriptbar\\
(\jscript\|\iscript )_y^2&=\sum _x(\jscript\|\iscript )_{xy}=\sum _x\iscript _x\circ\jscript _y=\iscriptbar\circ\jscript _y=(\jscript\vert\iscriptbar\,)_y
\end{align*}
Hence, $\iscript _x\circ\jscriptbar$ and $(\jscript\vert\iscriptbar\,)_y$ coexist with joint bi-instrument $(\jscript\|\iscript )$. In a similar way, if $B\in\rmob (H_2)$, we define 
$B$ \textit{given} $\iscript$ to be the bi-observable $(B\|\iscript )\in\rmob (H_1)$ defined by $(B\|\iscript )_{xy}=\iscript _x^*(B_y)$. The marginals of $(B\|\iscript )$ become
\begin{align*}
(B\|\iscript )_x^1&=\sum _y(B\|\iscript )_{xy}=\sum _y\iscript _x^*(B_y)=\iscript _x^*(I_{H_2})=\iscripthat _x\\
(B\|\iscript )_y^2&=\sum _x(B\|\iscript )_{xy}=\sum _x\iscript _x^*(B_y)=\iscriptbar\,^*(B_y)=(B\vert\iscript )_y
\end{align*}
We conculde that $\iscripthat$ and $(B\vert\iscript )$ coexist with joint bi-observable $(B\|\iscript )$.

\begin{example}  
A \textit{Holevo instrument} $\hscript ^{(A,\alpha )}\in\rmin (H_1,H_2)$ has the form \cite{hol82,hol98}
\begin{equation*}
\hscript _x^{(A,\alpha )}(\rho )=\trace (\rho A_x)\alpha _x
\end{equation*}
where $A\in\rmob (H_1)$, $\alpha =\brac{\alpha _x\colon x\in\Omega A}\subseteq\sscript (H_2)$. Then $\hscript _x^{(A,\alpha )*}(b)=\trace (\alpha _xb)A_x$ for all
$b\in\escript (H_2)$ and
\begin{equation*}
\hscripthat _x^{(A,\alpha )}=\hscript ^{(A,\alpha )*}(I_{H_2})=A_x
\end{equation*}
so $\hscripthat ^{(A,\alpha )}=A$. If $B\in\rmob (H_2)$, then
\begin{equation*}
\paren{B\|\hscript ^{(A,\alpha )}}_{xy}=\hscript _x^{(A,\alpha )*}(B_y)=\trace(\alpha _xB_y)A_x
\end{equation*}
The marginals become
\begin{align*}
\paren{B\|\hscript ^{(A,\alpha )}}^1&=\hscripthat\,^{(A,\alpha )}=A\\
\paren{B\|\jscript ^{(A,\alpha )}}^2&=(B\vert\hscript ^{(A,\alpha )})
\end{align*}
Hence, $A$ and $(B\vert\hscript ^{(A,\alpha )})$ coexist with joint bi-observable $(B\|\hscript ^{(A,\alpha )})$. Let $\hscript ^{(A,\alpha )}\in\rmin (H_1,H_2)$ and
$\hscript ^{(B,\beta )}\in\rmin (H_2,H_3)$. We then obtain
\begin{align*}
\paren{\hscript ^{(B,\beta )}\|\hscript ^{(A,\alpha )}}_{xy}(\rho )&=\hscript _y^{(B,\beta )}\paren{\hscript _x^{(A,\alpha )}(\rho )}
   =\hscript _y^{(B,\beta )}\sqbrac{\trace (\rho A_x)\alpha _x}\\
   &=\trace (\rho A_x)\hscript _y^{(B,\beta )}(\alpha _x)=\trace (\rho A_x)\trace (\alpha _xB_y)\beta _y\\
   &=\trace\sqbrac{\rho\trace (\alpha _xB_y)A_x}\beta _y=\trace\sqbrac{\rho C_{xy}}\delta _{xy}\\
   &=\hscript _{xy}^{(C,\delta )}(\rho )
\end{align*}
where $C_{xy}\in\rmob (H_1)$ is given by $C_{xy}=\trace (\alpha _xB_y)A_x$ and $\delta _{xy}=\beta _y$ for all $x\in\Omega _A$. We then have the marginals
\begin{align*}
\paren{\hscript ^{(B,\beta )}\|\hscript ^{(A,\alpha )}}_x^1(\rho )&=\sum _y\hscript _{xy}^{(C,\delta )}(\rho )=\trace (\rho A_x)
   =\trace (\rho A_x)\sum _y\trace (\alpha _xB_y)\beta _y\\
\intertext{and}
\paren{\hscript ^{(B,\beta )}\|\hscript ^{(A,\alpha )}}_y^2(\rho )&=\sqbrac{\sum _x\trace (\rho A_x)\trace (\alpha _xB_y)}\beta _y
\end{align*}
Moreover, we have
\begin{equation*}
\paren{\hscript ^{(B,\beta )}\vert\hscriptbar\,^{(A,\alpha )}}_y=\paren{\hscript ^{(B,\beta )}\|\hscript ^{(A,\alpha )}}_y^2
\end{equation*}
This completes the example.\hfill\qedsymbol
\end{example}

\section{Measurement Models}  
A \textit{measurement model} is a 4-tuple $M=(H,K,\iscript ,P)$ where $H$ and $K$ are Hilbert space $\iscript\in\rmin (H,H\otimes K)$ and $P\in\rmob (K)$
\cite{gud223,hz12}. The motivation behind $M$ is the following. In order to obtain information about the \textit{base space} $H$ we employ an \textit{auxiliary space} $K$ and an \textit{interaction instrument} $\iscript\in\rmin(H,H\otimes K)$. After the interaction is performed the \textit{probe observable} $P$ is measured and the resulting outcome gives the desired information. The \textit{bi-instrument measured by} $M$ is $\jscript\in\rmin (H)$ defined as
\begin{equation*}
\jscript _{xy}(\rho )=\trace _K\sqbrac{\iscript _x(\rho )I_H\otimes P_y}
\end{equation*}
for all $\rho\in\sscript (H)$, $x\in\Omega _\iscript$, $y\in\Omega _P$ where $\trace _K$ is the partial trace with respect to $K$. The \textit{instrument measured by} $M$ is the marginal
\begin{equation*}
\jscript _y^2(\rho )=\trace _K\sqbrac{\iscriptbar (\rho )I_H\otimes P_y}
\end{equation*}
The other marginal $\jscript _x^1(\rho )=\trace _K\sqbrac{\iscript _x(\rho )}$ is independent of the probe observable $P$ and is called the instrument $\iscript$
\textit{reduced to} $H$. The \textit{bi-observable measured by} $M$ is $\jscripthat\in\rmob (H)$ defined as $\jscripthat _{xy}=\jscript _{xy}^*(I_{H\otimes K})$. We have
\begin{equation*}
\trace (\rho\jscripthat _{xy})=\trace\sqbrac{\rho\jscript _{xy}^*(I_{H\otimes K})}=\trace\sqbrac{\jscript _{xy}(\rho )}=\trace\sqbrac{\iscript _x(\rho )I_H\otimes P_y}
\end{equation*}
for all $\rho\in\sscript (H)$. The \textit{observable measured by} $M$ is the marginal $\jscripthat\,^2=\sum _x\jscript _{xy}^*(I_{H\otimes K})$ and for all $\rho\in\sscript (H)$ we have
\begin{equation*}
\trace (\rho\jscripthat _y^2)=\trace\sqbrac{\jscript _y^2(\rho )}=\trace\sqbrac{\,\jscriptbar (\rho )I_{H\otimes K}}
\end{equation*}
The other marginal is $\jscripthat _x^1=\sum\limits _y\jscript _{xy}^*(I_{H\otimes K})$ and for all $\rho\in\sscript (H)$ we obtain
\begin{equation*}
\trace (\rho\jscripthat _x^1)=\trace\sqbrac{\jscript _x^1(\rho )}=\trace\sqbrac{\iscript _x(\rho )}=\trace\sqbrac{\rho\iscript _x^*(I_{H\otimes K}}
\end{equation*}
Hence, $\jscripthat _x^1=\iscript _x^*(I_{H\otimes K})=\iscripthat _x$ which again is independent of $P$. Since
\begin{equation*}
\trace (\rho\jscripthat _y^2)=\trace\brac{\trace _K\sqbrac{\,\iscriptbar (\rho )I_{H\otimes K}}}=\trace\sqbrac{\jscript _y^2 (\rho )}=\trace (\rho\jscript _y^{2\wedge})
\end{equation*}
we conclude that $\jscripthat _y^2=\jscript _y^{2\wedge}$.

Following the discussion in the previous paragraph, the main quantities of interest are $\jscript _y^2$ and $\jscripthat _y^2$. Suppose $\iscriptbar$ has Kraus decomposition $\iscriptbar (\rho )=\sum K_i\rho K_i^*$ where $K_i\in\lscript (H,H\otimes K)$, $i=1,2,\ldots ,n$ satisfies $\sum K_i^*K_i=I_H$. Then
\begin{align*}
\jscript _y^2(\rho )&=\trace _K\paren{\sum _iK_i\rho K_i^*(I_H\otimes P_y)}=\sum _i\trace _K(K_i\rho K_i^*I_H\otimes P_y)
\intertext{and}
\trace (\rho\jscripthat _y^2)&=\trace\sqbrac{\jscript _y^2(\rho )}=\sum _i\trace (K_i\rho K_i^*I_H\otimes P_y)=\sum _i\trace (\rho K_i^*I_H\otimes P_yK_i)\\
   &=\trace\paren{\rho\sum _iK_i^*I_H\otimes P_yK_i}
\end{align*}
Hence,
\begin{equation*}
\jscripthat _y^2=\jscript _y^{2\wedge}=\sum _iK_i^*I_H\otimes P_yK_i
\end{equation*}

In general, $\jscript ^2$ and $\jscripthat ^2$ are difficult to calculate, so we introduce a simplification called separability. \cite{gud223,hz12}. We say that a channel
$\iscript\colon\sscript (H)\to\sscript (H\otimes K)$ is \textit{Kraus separable} if there exist operators $K_i\in\lscript (H)$ and states $\rho _i\in\sscript (H)$, $i=1,2,\ldots ,n$, with $\sum\limits K_i^*K_i= I_H$ such that $\iscript (\rho )=\sum _i(K_i\rho K_i^*\otimes\rho _i)$ for all $\rho\in\sscript (H)$.

\begin{thm}    
\label{thm31}
Let $M=(H,K,\iscript ,P)$ be a measurement model with $\iscriptbar$ Kraus separable.
{\rm{(i)}}\enspace For $a\in\escript (H)$, $b\in\escript (K)$ we have $\iscriptbar\,^*(a\otimes b)=\sum\limits _i\trace (\rho _ib)K_i^*aK_i$.
{\rm{(ii)}}\enspace For all $y\in\Omega _P$, $\rho\in\sscript (H)$ we have $\jscript _y^2(\rho )=\sum\limits _i\trace (\rho _iP_y)K_i\rho K_i^*$.
{\rm{(iii)}}\enspace For all $y\in\Omega _P$ we obtain
\begin{equation*}
\jscripthat _y^2=\jscript _y^{2\wedge}=\sum _i\trace (\rho _iP_y)K_i^*K_i
\end{equation*}
which is a post-processing of the observable $\brac{K_i^*K_i:i=1,2,\ldots ,n}$.
\end{thm}
\begin{proof}
(i)\enspace For all $\rho\in\sscript (H)$ we obtain
\begin{align*}
\trace\sqbrac{\rho\iscriptbar ^*(a\otimes b)}&=\trace\sqbrac{\,\iscriptbar (\rho )a\otimes b}=\trace\sqbrac{\sum _i(K_i\rho K_i^*\otimes\rho _i)a\otimes b}\\
   &=\sum _i\trace (K_i\rho K_i^*a\otimes\rho _ib)=\sum _i\trace (K_i\rho K_i^*a)\trace (\rho _ib)\\
   &=\sum _i\trace (\rho _ib)\trace (\rho K_i^*aK_i)=\trace\sqbrac{\rho\sum _i\trace (\rho _ib)K_i^*aK_i}
\end{align*}
Hence, $\iscriptbar ^*(a\otimes b)=\sum\limits _i\trace (\rho _ib)K_i^*aK_i$.\newline
(ii)\enspace We have for all $y\in\Omega _P$, $\rho\in\sscript (H)$ that
\begin{align*}
\jscript _y^2(\rho )&=\trace _K\sqbrac{\,\iscriptbar (\rho )I_H\otimes P_y}=\trace _K\sqbrac{\sum _i(K_i\rho K_i^*\otimes\rho _i)(I_H\otimes P_y)}\\
   &=\sum _i\trace _K(K_i\rho K_i^*\otimes\rho _iP_y)=\sum _i\trace (\rho _iP_y)K_i\rho K_i^*
\end{align*}
(iii) Applying (i) we obtain
\begin{equation*}
\trace (\rho\jscripthat _y^2)=\trace\sqbrac{\,\iscriptbar (\rho )I_H\otimes P_y}=\trace\sqbrac{\rho\iscriptbar\,^*(I_H\otimes P_y)}
   =\trace\sqbrac{\rho\sum _i\trace (\rho _iP_y)K_i^*K_i}
\end{equation*}
Hence, $\jscripthat _y^2=\sum\limits _i\trace (\rho _iP_y)K_i^*K_i$ which is a post-processing of the observable $\brac{K_i^*K_i\colon i=1,2,\ldots ,n}$ because
$\sum\limits _y\trace (\rho _iP_y)=\trace (\rho _i)=1$, $i=1,2,\ldots ,n$.
\end{proof}

A simple example of a Kraus separable channel has the form $\iscript (\rho )=\sum K_i\rho K_i^*$ where $K_i\in\lscript (H,H\otimes K)$ is given by
$K_i\phi = A_i\phi\otimes\psi _i$, $A_i\in\lscript (H)$, $\psi _i\in K$, $i=1,2,\ldots ,n$, with $\doubleab{\psi _i}=1$, $\sum A_i^*A_i=I_H$. The adjoint of $K_i$ satisfies
\begin{align*}
\elbows{\phi ,K_i^*(\phi _1\otimes\phi _2)}&=\elbows{K_i\phi ,\phi _1\otimes\phi _2}=\elbows{A_i\phi\otimes\psi _i,\phi _1\otimes\phi _2}\\
   &=\elbows{A_i\phi ,\phi _1}\elbows{\psi _i,\phi _2}=\elbows{\phi ,\elbows{\psi _i,\phi _2}A_i^*\phi _1}
\end{align*}
Hence, $K_i^*(\phi _1\otimes\phi _2)=\elbows{\psi _i,\phi _2}A_i^*\phi _1$ for all $\phi _1\in H$, $\phi _2\in K$, $i=1,2,\ldots ,n$. We conclude that for all
$\phi _1\otimes\phi _2\in H\otimes K$ we obtain
\begin{align*}
\iscript (\rho )(\phi _1\otimes\phi _2)&=\sum _iK_i\rho K_i^*(\phi _1\otimes\phi _2)=\sum _iK_i\rho\elbows{\psi _i,\phi _2}A_i^*\phi _1\\
   &=\sum _i\elbows{\psi _i,\phi _2}K_i\rho A_i^*\phi _1=\sum _i\elbows{\psi _i,\phi _2}A_i\rho A_i^*\phi _1\otimes\psi _i\\
   &=\sum _i\sqbrac{A_i\rho A_i^*\phi _1\otimes\elbows{\psi _i,\phi _2}\psi _i}=\sum _i\sqbrac{A_i\rho A_i^*\otimes\ket{\psi _i}\bra{\psi _i}}\phi _1\otimes\phi _2
\end{align*}
Thus, $\iscript (\rho )=\sum\limits _i\paren{A_i\rho A_i^*\otimes\ket{\psi _i}\bra{\psi _i}}$ where $\ket{\psi _i}\bra{\psi _i}$ are pure states. We conclude that $\iscript$ is Kraus separable with corresponding observable $\brac{A_i\colon i=1,2,\ldots ,n}$ and pure states $\rho _i=\ket{\psi _i}\bra{\psi _i}$.

Let $\iscript =\hscript ^{(A,\alpha )}\in\rmin (H,H\otimes K)$ be a Holevo instrument model $M=(H,K,\iscript ,P)$. The bi-instrument measured by $M$ becomes
\begin{equation*}
\jscript _{xy}(\rho )=\trace _K\sqbrac{\iscript _x(\rho )I_H\otimes P_y}=\trace _K\sqbrac{\trace (\rho A_x)\alpha _xI_H\otimes P_y}
    =\trace (\rho A_x)\trace _K\sqbrac{\alpha _xI_H\otimes P_y}
\end{equation*}
The instrument measured by $M$ satisfies
\begin{equation*}
\jscript _y^2(\rho )=\sum _x\jscript _{xy}(\rho )=\sum _x\trace (\rho A_x)\trace _K(\alpha _xI_H\otimes P_y)
\end{equation*}
The other marginal becomes
\begin{equation*}
\jscript _x^1(\rho )=\trace _K\sqbrac{\iscript _x(\rho )}=\trace _K\sqbrac{\trace (\rho A_x)\alpha _x}=\trace (\rho A_x)\trace _K(\alpha _x)
\end{equation*}
Since
\begin{equation*}
\trace\sqbrac{\rho\jscript _x^{1*}(I_H)}=\trace\sqbrac{\jscript _x^1(\rho )}=\trace (\rho A_x)
\end{equation*}
we have $\jscript _x^{1\wedge}=\jscript _x^{1*}(I_H)=A_x$. Since $\jscript ^{2\wedge}$ is complicated to calculate, we again introduce a separability condition. We say that
$\hscript ^{(A,\alpha )}$ is \textit{Holevo separable} if $\alpha _x=\beta _x\otimes\gamma _x$ where $\beta _x\in\sscript (H)$, $\gamma _x\in\sscript (K)$. The next result is similar to Theorem~\ref{thm31}.

\begin{thm}    
\label{thm32}
Suppose $\iscript =\hscript ^{(A,\alpha )}$ is Holevo separable.
{\rm{(i)}}\enspace For all $a\in\escript (H\otimes K)$ we have $\iscript _x^*(a)=\trace (\beta _x\otimes\gamma _xa)A_x$.
{\rm{(ii)}}\enspace $\jscript _{xy}(\rho )=\trace (\rho A_x)\trace (\gamma _xP_y)\beta _x$.
{\rm{(iii)}}\enspace $\jscript _y^2(\rho )=\sum _x\trace (\rho A_x)\trace (\gamma _xP_y)\beta _x$.
{\rm{(iv)}}\enspace $\jscript _x^1=\hscript _x^{(A,\beta )}$.
{\rm{(v)}}\enspace $\jscripthat _{xy}=\trace (\gamma _xP_y)A_x$.
{\rm{(vi)}}\enspace $\jscripthat _y^2=\jscript _y^{2\wedge}=\sum _x\trace (\gamma _xP_x)A_x$ which is a post-processing of $A$.
\end{thm}
\begin{proof}
(i)\enspace For all $\rho\in\sscript (H)$ we have
\begin{align*}
\trace\sqbrac{\rho\iscript _x^*(a)}&=\trace\sqbrac{\iscript _x(\rho )a}=\trace\sqbrac{\hscript _x^{(A,\alpha )}(\rho )a}=\trace\sqbrac{\trace (\rho A_x)\alpha _xa}\\
   &=\trace (\rho A_x)\trace (\alpha _xa)=\trace (\rho A_x)\trace (\beta _x\otimes\gamma _xa)=\trace\sqbrac{\rho\trace (\beta _x\otimes\gamma _xa)A_x}
\end{align*}
Hence, $\iscript _x^*(a)=\trace (\beta _x\otimes\gamma _xa)A_x$.
(ii)\enspace We have that
\begin{align*}
\jscript _{xy}(\rho )&=\trace _K\sqbrac{\iscript _x(\rho )I_H\otimes P_y}=\trace _K\sqbrac{\hscript _x^{(A,\alpha )}(\rho )I_H\otimes P_y}
   =\trace _K\sqbrac{\trace (\rho A_x)\alpha _xI_H\otimes P_y}\\
   &=\trace _K\sqbrac{\trace (\rho A_x)\beta _x\otimes\gamma _xI_H\otimes P_y}=\trace (\rho A_x)\trace _K(\beta _x\otimes\gamma _xP_y)\\
   &=\trace (\rho A_x)\trace (\gamma _xP_y)\beta _x
\end{align*}
(iii)\enspace Since $\iscriptbar (\rho )=\hscriptbar\,^{(A,\alpha )}(\rho )=\sum _x\trace (\rho A_x)\alpha _x-\sum _x\trace (\rho A_x)\beta _x\otimes\gamma _x$ we have
\begin{align*}
\jscript _y^2(\rho )&=\trace _K\sqbrac{\,\iscriptbar (\rho )I_H\otimes P_y}=\trace _K\sqbrac{\sum _x\trace (\rho A_y)\beta _x\otimes\gamma _xI_H\otimes P_y}\\
   &=\sum _x\trace (\rho A_x)\trace _K(\beta _x\otimes\gamma _xP_y)=\sum _x\trace (\rho A_x)\trace (\gamma _xP_y)\beta _x
\end{align*}
(iv)\enspace For all $\rho\in\sscript (H)$ we obtain
\begin{align*}
\jscript _x^1(\rho )&=\trace _K\sqbrac{\iscript _x(\rho )}=\trace _K\sqbrac{\hscript _X^{(A,\alpha )}(\rho )}=\trace _K\sqbrac{\trace (A_x)\alpha _x}\\
   &=\trace (\rho A_x)\trace _K(\beta _x\otimes\gamma _x)=\trace (\rho A_x)\beta _x=\hscript _x^{(A,\beta )}(\rho )
\end{align*}
Hence, $\jscript _x^1=\hscript _x^{(A,\beta )}$.
(v)\enspace By (ii) we have
\begin{equation*}
\trace\sqbrac{\rho\jscript _{xy}^*(I_{H\otimes K})}=\trace\sqbrac{\jscript _{xy}(\rho )}=\trace (\rho A_x)\trace (\gamma _xP_y)
   =\trace\sqbrac{\rho\trace (\gamma _xP_y)A_x}
\end{equation*}
Therefore, $\jscript _{xy}^*(I_{H\otimes K})=\trace (\gamma _xP_y)A_x$ so we obtain
\begin{equation*}
\jscripthat _{xy}=\jscript _{xy}^*(I_{H\otimes K})=\trace (\gamma _xP_y)A_x
\end{equation*}
(vi)\enspace Applying (v) gives
\begin{equation*}
\jscripthat _y^2=\jscript _y^{2\wedge}=\sum _x\jscript _{xy}^*(I_{H\otimes K})=\sum _x\trace (\gamma _xP_y)A_x
\end{equation*}
This is a post-processing of $A$ because $\sum _y\trace (\gamma _xP_y)=1$ for all $x\in\Omega _A$.
\end{proof}

\end{document}